\documentclass[letterpaper, 10 pt, conference]{ieeeconf}  

\usepackage[T1]{fontenc}
\usepackage[utf8]{inputenc}

\renewcommand\labelenumi{(\roman{enumi})}
\renewcommand\theenumi\labelenumi

\usepackage{amsmath}
\usepackage{amsfonts}
\usepackage[exponent-product=\cdot, per-mode=symbol]{siunitx}

\usepackage{booktabs}
\usepackage{graphicx}
\usepackage{epstopdf}
\usepackage{xcolor}
\definecolor{vgRed}{RGB}{193, 48, 24}
\definecolor{vgOrange}{RGB}{243, 111, 19}
\definecolor{vgYellow}{RGB}{235, 203, 56}
\definecolor{vgGreen}{RGB}{162, 185, 105}
\definecolor{vgLightBlue}{RGB}{13, 149, 188}
\definecolor{vgDarkBlue}{RGB}{6, 56, 81}

\newcommand{\N}{\mathbb{N}}

\newcommand{\Nz}{\mathbb{N}_0}
\newcommand{\R}{\mathbb{R}}
\newcommand{\Rp}{\R^+}
\newcommand{\Rpz}{\Rp_0}
\newcommand{\dt}{\text{d}t}
\newcommand{\Tr}{\operatorname{Tr}}
\newcommand{\LUT}{\operatorname{LUT}}

\usepackage{acronym}
\acrodef{bem}[BEM]{blade element momentum theory}
\acrodef{lq}[LQ]{linear quadratic}
\acrodef{mpc}[MPC]{Model Predictive Controller}
\acrodef{wt}[WT]{Wind Turbine}
\acrodef{pi}[PI]{proportional-integral}
\acrodef{siso}[SISO]{single-input single-output}
\acrodef{lmi}[LMI]{linear matrix inequality}
\acrodefplural{lmi}[LMIs]{linear matrix inequalities}
\acrodef{del}[DEL]{damage equivalent load}
\acrodef{rms}[RMS]{root-mean-square}
\acrodef{ti}[TI]{turbulence intensity}

\usepackage{tikz}
\usepackage{pgfplots}
\usepackage{pgfplotstable}
\pgfplotsset{compat=newest}
\usetikzlibrary{plotmarks}
\usetikzlibrary{fit}
\usetikzlibrary{positioning}
\usetikzlibrary{shapes}
\usetikzlibrary{arrows,automata,plotmarks}
\usetikzlibrary{decorations.pathreplacing,decorations.markings}
\usetikzlibrary{backgrounds}
\usetikzlibrary{calc}
\usetikzlibrary{patterns}
\usetikzlibrary{matrix}
\usepgfplotslibrary{statistics}
\usepgfplotslibrary{groupplots}
\usepgfplotslibrary{external}
\tikzset{external/system call={pdflatex \tikzexternalcheckshellescape -halt-on-error
    -interaction=batchmode -jobname "\image" "\texsource"}}
\newtheorem{problem}{Problem}
\newtheorem{lemma}{Lemma}

\usepackage{todonotes}
\renewcommand{\todo}[2][]{\tikzexternaldisable\@todo[#1]{#2}\tikzexternalenable}

\pgfplotsset{every axis/.append style={semithick,tick style={major tick
            length=4pt,semithick,black}}}

\pgfkeys{/pgfplots/x axis shift down/.style={
        x axis line style={yshift=-#1},
        xtick style={yshift=-#1},
        tick align=outside,
        xticklabel shift={#1}}}

\pgfkeys{/pgfplots/y axis shift left/.style={
        y axis line style={xshift=-#1},
        ytick style={xshift=-#1},
        yticklabel shift={#1},
        scaled y ticks = false,
        tick align=outside,
        y tick label style={/pgf/number format/fixed,
        }
    }
}

\pgfplotsset{myPlot/.style={%
        width=8cm,
        height=4cm,
        line width = 0.7pt,
        separate axis lines,
        axis x line*=bottom,
        x axis shift down = 3pt,
        enlarge x limits=false,
        axis y line*=left,
        y axis shift left = 6pt,
        enlarge y limits={abs=.25pt},
        enlarge x limits={abs=.25pt},
    }
}
\usepackage{mathtools,stackengine,scalerel}
\newcommand{\rawsat}[3]{\ThisStyle{\raisebox{#1\LMpt}{\kern.5\LMpt\scaleto{\rawsatimg{#3}}{#2\LMex}\kern.5\LMpt}}}
\newcommand{\rawsatimg}[1]{%
\tikzexternaldisable
\begin{tikzpicture}
\coordinate (A) at (-7,-7);
\coordinate (B) at (-2,-7);
\coordinate (C) at (2,7);
\coordinate (D) at (7,7);
\draw [black, line width=#1pt] (A)--(B)--(C)--(D);
\end{tikzpicture}%
\tikzexternalenable
}
\makeatletter
\newcommand\sat{
    \relax\if@display
        \mathop{\rawsat{-5}{3.2}{20}}
    \else
        \mathop{\rawsat{-2.4}{2.25}{30}}
    \fi
}

\pgfdeclareshape{satnode}{
\inheritsavedanchors[from={rectangle}]
\inheritbackgroundpath[from={rectangle}]
\inheritanchorborder[from={rectangle}]
\foreach \x in {center,north east,north west,north,south,south east,south west, west, east}{
\inheritanchor[from={rectangle}]{\x}
}
\foregroundpath{
\pgfpointdiff{\northeast}{\southwest}
\pgf@xa=\pgf@x \pgf@ya=\pgf@y
\northeast
\pgfpathmoveto{\pgfpoint{0}{0.45\pgf@ya}}
\pgfpathlineto{\pgfpoint{0}{-0.45\pgf@ya}}
\pgfpathmoveto{\pgfpoint{0.45\pgf@xa}{0}}
\pgfpathlineto{\pgfpoint{-0.45\pgf@xa}{0}}
\pgfpathmoveto{\pgfpointadd{\southwest}{\pgfpoint{-0.2\pgf@xa}{-0.3\pgf@ya}}}
\pgfpathlineto{\pgfpointadd{\southwest}{\pgfpoint{-0.5\pgf@xa}{-0.3\pgf@ya}}}
\pgfpathlineto{\pgfpointadd{\northeast}{\pgfpoint{-0.5\pgf@xa}{-0.3\pgf@ya}}}
\pgfpathlineto{\pgfpointadd{\northeast}{\pgfpoint{-0.4\pgf@xa}{-0.3\pgf@ya}}}
{
   \pgftransformshift{\pgfpointadd{\northeast}{\pgfpoint{-0.4\pgf@xa}{-0.3\pgf@ya}}}
   \pgftransformscale{0.5}
   \pgfsetcolor{black}
}
\pgfusepath{stroke}
}
}

\pgfdeclareshape{lowpassnode}{
\inheritsavedanchors[from={rectangle}]
\inheritbackgroundpath[from={rectangle}]
\inheritanchorborder[from={rectangle}]
\foreach \x in {center,north east,north west,north,south,south east,south west, west, east}{
\inheritanchor[from={rectangle}]{\x}
}
\foregroundpath{
\pgfpointdiff{\northeast}{\southwest}
\pgf@xa=\pgf@x \pgf@ya=\pgf@y
\northeast
\pgfpathmoveto{\pgfpoint{0.4\pgf@xa}{-0.2\pgf@ya}}
\pgfpathlineto{\pgfpoint{0}{-0.2\pgf@ya}}
\pgfpathlineto{\pgfpoint{-0.4\pgf@xa}{0.2\pgf@ya}}
\pgfusepath{stroke}
}
}

\makeatother

\IEEEoverridecommandlockouts                              

\title{\LARGE \bf
Robust LQ Optimal Control for Wind Turbine Power Tracking Operation
}

\author{Aaron Grapentin$^1$, Christian A. Hans$^2$, and J\"org Raisch$^3$
\thanks{*This work was partially supported by the German Federal Ministry for Economic Affairs and and Climate Action (BMWK), project no. 03EE2036C.}
\thanks{$^{1}$Aaron Grapentin is with the Control Systems Group, Technische Universit\"at Berlin, Germany, {\tt\small grapentin@control.tu-berlin.de}}
\thanks{$^{2}$Christian A. Hans is with the Automation and Sensorics in Networked Systems Group, University of Kassel, Germany, {\tt\small hans@uni-kassel.de}}
\thanks{$^{3}$J\"org Raisch is with the Control Systems Group, Technische Universit\"at Berlin, Germany and with Science of Intelligence, Research Cluster of Excellence, Berlin, Germany, {\tt\small raisch@control.tu-berlin.de}}%
\thanks{We thank Arnold Sterle for fruitful discussions and insightful comments.}%
}

\begin{document}

\maketitle
\thispagestyle{empty}
\pagestyle{empty}

\begin{abstract}

In this paper, a robust linear quadratic optimal control approach for accurate active power tracking of wind turbines is presented.
For control synthesis, linear matrix inequalities are employed using an augmented wind turbine state model with uncertain parameters.
The resulting controller ensures robust stability in different operating regions.
In a case study, the novel approach is compared to existing controllers from literature.
Simulations indicate that the controller improves power tracking accuracy while leading to similar mechanical wear as existing approaches.

\end{abstract}
\acresetall


\section{Introduction}
Renewable energy sources play a central role in the fight against climate change.
Especially, solar and wind power have experienced a fast growth, which is expected to accelerate further \cite{murdock2021renewables}.
The rise of wind power in certain markets has lead to new requirements for wind farm operators, such as active power control \cite{de2007connection}:
to avoid undesired grid conditions, wind farms must curtail their production at certain times.
In such cases, the wind turbines of a farm must track a desired power trajectory which often lies below the weather-dependent available infeed.
In Germany, regulations require a maximal power tracking deviation of \SI{5}{\percent} from a wind farm's rated power \cite{vde2018fnn}.
With increasing share of renewable energy sources, accurate power tracking beyond this requirement is desirable as it allows to reduce safety margins and improve the overall system performance.

Several approaches for power tracking of individual wind turbines exist.
In \cite{karimpour2021exact}, exact output regulation is achieved while employing LiDAR information.
In \cite{pham2012lqr}, a \ac{lq} optimal control approach for wind turbines is proposed and numerically analyzed.
While relying on wind estimates, this control approach improves rotor speed and output power tracking only marginally compared to the baseline controller from \cite{nam2011feedforward}.
In \cite{radaideh2021active}, active and reactive power of doubly-fed induction generators are controlled.
While disturbance rejection for uncertain wind speed is validated, mechanical loads on the wind turbine components are not studied.
In \cite{grapentin2022lq}, an \ac{lq} optimal controller is presented that reduces \aclp{del} while achieving power tracking using wind speed estimates.

While these approaches provide stability guarantees in single operating points, none of them can provide guarantees for sets of operating conditions.
Additionally, most approaches are combinations of heuristically designed single-input single-output controllers.

In this paper, we derive a multivarible \ac{lq} optimal control scheme that considers rotor as well as tower top fore-aft dynamics and provides accurate power tracking.
As is customary in wind turbine control, we design different controllers that are used for different operating regions, which are then combined into an overall controller.
Control synthesis for each convex operating region uses \ac{lmi} and ensures robust stability for multiple operating points within that operating region.
Practical applicability is increased by only relying on wind speed estimates obtained by an observer.
In a case study, the resulting controller is compared to state-of-the-art approaches.
Extensive closed-loop simulations indicate significant improvements in power tracking accuracy while keeping mechanical loads sufficiently small.

The remainder of this paper is structured as follows.
In Section~\ref{sec:model}, the plant model is introduced.
In Section~\ref{sec:controller}, our novel control scheme is presented.
In Section~\ref{sec:case_study}, the controller is evaluated in a case study.
Section~\ref{sec:conclusion} concludes this work.

\subsection{Notation}
The set of real numbers is denoted by $\R$, the set of positive real numbers by $\Rp$, the set of nonnegative real numbers by $\Rpz$ and the set of positive integers by $\mathbb{N}$.
The Euclidean norm is denoted by $\|\cdot\|_2$.
The trace operator $\Tr(\cdot)$ returns the sum of all diagonal entries of a given square matrix.
The operator $\LUT(\cdot)$ denotes a lookup-table using linear interpolation.
The saturation operator $\sat$ represents
\begin{align}
y = \sat_{\underline{x}}^{\overline{x}}x
=\begin{cases}
\underline{x}, &\text{if } x \leq \underline{x},\\
\overline{x}, &\text{if } x \geq \overline{x},\\
x, &\text{else},
\end{cases}
\end{align}
where $x\in\R$ is being saturated by the bounds $\underline{x},\overline{x}\in\R$, with $\underline{x}<\overline{x}$, such that $y\in[\underline{x}, \overline{x}]$.
Negative definiteness of a square matrix $Q$ is denoted by $Q\prec 0$.
Analogously, ${Q\succ0}$ denotes positive definiteness of $Q$.
Negative and positive semi-definiteness of $Q$ are denoted by $Q\preceq 0$ and $Q\succeq 0$.
Moreover, $Q_1\succ Q_2$ denotes that $Q_1 - Q_2$ is positive definite.
For any $j\in\mathbb{N}$, $I_j$ represents the $j\times j$ identity matrix.
With $n,m\in\N$, a zero matrix of size $n\times m$ is denoted by $0_{n\times m}$.


\section{Model}
\label{sec:model}
Our model of a wind turbine consists of three parts: the aerodynamic system models the power captured by the rotor.
The drive train and generator dynamics model the conversion of mechanical power at the rotor to electrical power provided to the grid.
Lastly, the fore-aft tower top dynamics model tower position and velocity.
In what follows, these different parts will be first discussed separately and then combined into an overall state model.
Throughout this paper, we assume that the wind direction is perpendicular to the rotor.

\subsection{Aerodynamic System}
At time $t\in\Rpz$, the wind power is given by
\begin{equation}
P_\text{wind}(t) = \frac{\rho}{2}\pi r^2 V(t)^3,
\label{eq:wind_power}
\end{equation}
where $r\in\Rp$ is the rotor radius, $\rho\in\Rp$ the air density, and $V(t)\in\Rp$ the wind speed \cite{Aho2012,Boersma2017}.
The power coefficient
\begin{equation}
C_p\big(\lambda(t), \theta(t)\big) = \frac{P_r(t)}{P_\text{wind}(t)},
\label{eq:power_coefficient_ratio}
\end{equation}
describes the ratio of rotor power $P_r(t)\in\Rpz$ to wind power.
Evidently, $P_r(t)\leq P_\text{wind}(t)$ and $C_p(\cdot)\in[0,1]$.
Apart from blade material, geometry, and weight, $C_p(\cdot)$ depends on the pitch angle $\theta(t)\in\R$ and the tip-speed ratio
\begin{equation}
\lambda(t)=r\frac{\omega_r(t)}{V(t)}\in\Rpz,
\label{eq:lambda_definition}
\end{equation}
which describes the relation of blade-tip speed at rotor angular velocity $\omega_r(t)\in\Rp$ and wind speed $V(t)$.
In this work, a piecewise-affine approximation of $C_P(\cdot)$ using data of the IEA \SI{3.4}{\mega\watt} land-based wind turbine \cite{RWT} is employed.
Combining (\ref{eq:wind_power}) and (\ref{eq:power_coefficient_ratio}) allows us to express the rotor torque as
\begin{equation}
M_r(t)=\frac{P_r(t)}{\omega_r(t)} =\frac{\rho}{2}\frac{\pi r^2 V(t)^3}{\omega_r(t)}C_p(\lambda(t), \theta(t)).
\label{eq:generator_torque_definition}
\end{equation}

\subsection{Drive Train and Generator}
We assume a stiff coupling of rotor, gearbox, and generator without any shaft deformation, i.e., the generator angular velocity is $\omega(t) = N_g\omega_r(t)$, where $N_g\in\Rp$ denotes the gearbox ratio.
The drive-train dynamics can be described by
\begin{equation}
J_t \dot{\omega}_r(t) = \frac{J_t}{N_g}\dot{\omega}(t) = M_r(t) - N_g M_g(t),
\label{eq:gearbox_definition}
\end{equation}
with moment of inertia $J_t\in\Rp$ and generator torque ${M_g(t)\in\Rpz}$ \cite{schlipf2013nonlinear}.
Moreover, the electrical power is
\begin{equation}
P(t)=\eta\omega(t)M_g(t),
\label{eq:generator_definition}
\end{equation}
with combined gearbox and generator efficiency $\eta\in(0,1]$.

\subsection{Tower Dynamics}
The fore-aft tower top position $x_t(t)\in\R$ and velocity $v_t(t)=\dot{x}_t(t)$ dynamics follow the equation
\begin{align}
    M_t\dot{v}_t(t)\! + D_t v_t(t)\! + K_t x_t(t) = \frac{\rho\pi r^2}{2}V(t)^2 C_T\big(\lambda(t), \theta(t)\big),
    \label{eq:tower_dynamics}
\end{align}
where $M_t, D_t, K_t\in\Rp$ denote tower top mass, damping, and stiffness, respectively.
Moreover, the right hand side of (\ref{eq:tower_dynamics}) describes the aerodynamic force onto the tower \cite{bottasso2014lidar,lescher2006multiobjective}.
Similar to the power coefficient, the thrust coefficient $C_T(\cdot)$ is modeled by a piecewise-affine approximation based on \cite{RWT}.

\subsection{State Model}
Combining (\ref{eq:generator_torque_definition}), (\ref{eq:gearbox_definition}), and (\ref{eq:tower_dynamics}) provides a first state model
\begin{align}
\begin{bmatrix}
\dot{\omega}(t) \\ \dot{x}_t(t) \\ \dot{v}_t(t)
\end{bmatrix}
\!\!\!=\!\!\!\begin{bmatrix}
\frac{\rho\pi r^2 N_g^2}{2 J_t}\frac{V(t)^3}{\omega(t)}C_P\big(\lambda(t), \theta(t)\big) - \frac{N_g^2}{J_t}M_g(t)\\
v_t(t)\\
\frac{\rho\pi r^2}{2M_t}V(t)^2C_T\big(\lambda(t), \theta(t)\big)\!-\!\frac{D_t}{M_t}v_t(t)\!-\!\frac{K_t}{M_t}x_t(t)
\end{bmatrix}.
\label{eq:model}
\end{align}
Here, $\begin{bmatrix}\theta(t)& M_g(t)\end{bmatrix}^T$ is determined by the vector of control inputs (see Section \ref{sec:augmentation}) and $V(t)$ is an uncertain external input.


\section{Control Approach}
\label{sec:controller}
In Figure \ref{fig:lq_overview}, an overview of the presented control approach is shown.
On the left hand side, the equilibrium $x^s$ of an augmented state model (see Section~\ref{sec:augmentation}) is found for a given estimate $\hat{w}$ of the uncertain external input vector $w$.
The latter includes the wind speed $V\in\R$, the desired generator speed $\omega^d\in\R$, and the desired power $P^d\in\R$.
For the specific definitions of $\omega^d$ and $P^d$, please refer to Section~\ref{sec:trajectory_generation}.
Based on the deviation of the state $x$ from the equilibrium $x^s$, control actions $\mu$ are deduced by the controller and applied with the equilibrium input $u^s$ to the augmented wind turbine model.
In addition to the control input $u$, the external input $w$ influences the system's behavior.
\begin{figure}[htbp]
	\centering
	\includegraphics{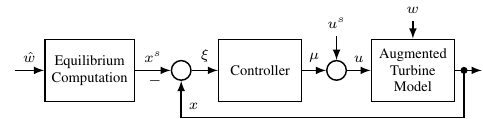}
	\caption{Closed-loop control scheme.}
	\label{fig:lq_overview}
\end{figure}

In what follows, we will design two controllers for different operating regions which will be combined later into one control law.
We first derive an augmented wind turbine model in Section~\ref{sec:augmentation}.
Then, we compute the equilibrium in Section~\ref{sec:equilibrium}.
In Section~\ref{sec:linearization}, the linearization is discussed.
The robust control design for each of the two operating region is presented in Section~\ref{sec:uncertainty}.
In Section~\ref{sec:trajectory_generation}, the derivation of the external input estimate $\hat{w}$ is discussed.
Lastly, Section~\ref{sec:switching_controller_subsection} presents the overall controller that combines the controllers designed in Section~\ref{sec:uncertainty}.

\subsection{State Model Augmentation}
\label{sec:augmentation}
We augment the state model as follows.
\begin{enumerate}
    \item We include the integral errors of generator speed and power as additional state variables $z_\omega(t)$ and $z_P(t)$
    \begin{subequations}
        \begin{align}
            \dot{z}_{\omega}(t)\!&=\!\omega^d(t)\!-\!\omega(t),\\
            \dot{z}_P(t)\!&=\!P^d(t)\!-\!P(t)\!=\!P^d(t)\!-\!\eta\omega(t)M_g(t).
            \label{eq:integrator_power_error}
        \end{align}
    \end{subequations}
    \item We use the derivatives of the pitch angle and generator torque, as control inputs $u_1$ and $u_2$, i.e.,
    \begin{subequations}
        \begin{align}
            \dot{\theta}(t) &= u_1(t),\\
            \dot{M}_g(t) &= u_2(t).
        \end{align}
    \end{subequations}
\end{enumerate}
The resulting augmented plant model can then be stated as
\begin{equation}
    \dot{x}(t)=f_a\big(x(t), w(t)\big)+B u(t),
    \label{eq:augmented_nonlinear_model}
\end{equation}
where 
\begin{subequations}
    \begin{align}
        x(t)&=\!\begin{bmatrix}\omega(t)&\!x_t(t)&\!v_t(t)&\!z_\omega(t)&\!z_P(t)&\!\theta(t)&\!M_g(t)\end{bmatrix}^T,\nonumber\\
        u(t)&=\!\begin{bmatrix}u_1(t)&\!u_2(t)\end{bmatrix}^T,\nonumber\\
        w(t)&=\!\begin{bmatrix}V(t)&\!\omega^d(t)&\!P^d(t)\end{bmatrix}^T,\nonumber
    \end{align}
\end{subequations}
are the augmented state vector, the vector of control inputs, and the vector of external inputs, respectively.
Moreover, 
\begin{align}
    &f_a\big(x(t),w(t)\big)=\nonumber\\
    &\begin{bmatrix}
        \frac{\rho\pi r^2 N_g^2}{2 J_t}\frac{V(t)^3}{\omega(t)}C_P\big(\lambda(t), \theta(t)\big) - \frac{N_g^2}{J_t}M_g(t)\\
        v_t(t)\\
        \frac{\rho\pi r^2}{2M_t}V(t)^2C_T\big(\lambda(t), \theta(t)\big)-\frac{D_t}{M_t}v_t(t)- \frac{K_t}{M_t}x_t(t)\\
        \omega^d(t) - \omega(t)\\
        P^d(t) - \eta\omega(t)M_g(t)\\
        0\\
        0 
    \end{bmatrix},
\end{align}
and $B = \begin{bmatrix}0_{2\times 5}& I_2\end{bmatrix}^T$.

\subsection{Equilibrium Computation}
\label{sec:equilibrium}
For the upcoming \ac{lq} optimal control design, linear state models are needed.
Therefore, we determine equilibria around which we will linearize in Section~\ref{sec:linearization}.
\begin{lemma}
Given a constant external input $w^s=\begin{bmatrix}V^s& \omega^{d,s}& P^{d,s}\end{bmatrix}^T$, there exists an equilibrium $(x^s, u^s, w^s)$ with
$x^s=\begin{bmatrix}\omega^s&x_t^s&v_t^s&z_{\omega}^s&z_P^s&\theta^s&M_g^s\end{bmatrix}^T$ and ${u^s=\begin{bmatrix}u_1^s& u_2^s\end{bmatrix}^T}$, if and only if
\begin{equation}
C_P(\lambda^s, \theta^s) = \frac{2 P^{d,s}}{\rho\pi r^2\eta (V^s)^3},
\label{eq:equilibrium_condition}
\end{equation}
with $\lambda^s=\frac{r}{N_g}\frac{\omega^{d,s}}{V^s}$, has a solution $\theta^s\in[\underline{\theta}, \overline{\theta}]$, where $\underline{\theta}, \overline{\theta}\in\R$ denote the lower and upper pitch angle bound.
\end{lemma}

\begin{proof}
Evidently, because of the particular structure of function $f_a$ and matrix $B$, the equilibrium condition
\begin{equation}
    f_a(x^s, w^s) + B u^s = 0,
\end{equation}
is equivalent to $u^s=0$ and $f_a(x^s, w^s)=0$.
Therefore, we have $v_t^s = 0$, $\omega^s = \omega^{d,s}$, and $M_g^s = \frac{P^{d,s}}{\eta\omega^{d,s}}$.
Employing $v^s_t = 0$, we further have that
\begin{equation}
x_t^s = \frac{\rho\pi r^2}{2 K_t}(V^s)^2 C_T(\lambda^s, \theta^s).
\end{equation}
Finally, $\dot{\omega}^s = 0$ can be reformulated into
\begin{equation}
C_P(\lambda^s, \theta^s) = \frac{2M_g^s\omega^{d,s}}{\rho\pi r^2 (V^s)^3} = \frac{2P^{d,s}}{\rho\pi r^2\eta (V^s)^3}.
\label{eq:cp_equilibrium}
\end{equation}
Here, $\lambda^s = \frac{r}{N_g}\frac{\omega^{d,s}}{V^s}$ is fixed via the external inputs $\omega^{d,s}$ and $V^s$.
Moreover, the right-hand side is fixed through $P^{d,s}$ and $V^s$.
Thus, if and only if we can find a $\theta^s\in[\underline{\theta}, \overline{\theta}]$ that satisfies (\ref{eq:cp_equilibrium}), then we can find an equilibrium.
\end{proof}
If, for a given $w^s$, multiple solutions $\theta^s$ of (\ref{eq:equilibrium_condition}) exist, the largest $\theta^s$ is chosen.

\subsection{Linearization}
\label{sec:linearization}
Linearizing (\ref{eq:augmented_nonlinear_model}) around an equilibrium $(x^s, u^s, w^s)$ leads to a state model of the form
\begin{equation}
    \dot{\xi}(t) = A\xi(t)+B\mu(t),
\end{equation}
with $\xi(t)\!=\!x(t)-x^s$ and $\mu(t)\!=\!u(t)-u^s\!=\!u(t)$.
Here, $A$ is the Jacobian associated with $f_a(\cdot)$ evaluated at $(x^s, w^s)$, which exhibits the following nonzero elements:

\begin{subequations}
    \begin{align}
        A^{(1,1)} &= \frac{\rho\pi r^2 N_g^2}{2 J_t}V(t)^{3}\left.\frac{\partial\frac{C_P\big(\lambda(t), \theta(t)\big)}{\omega(t)}}{\partial\omega(t)}\right\rvert_{(x^s, w^s)},
        \label{eq:derivative_ai_11}\\
        A^{(1,6)} &= \frac{\rho\pi r^2 N_g^2}{2 J_t}\frac{V(t)^3}{\omega(t)}\left.\frac{\partial C_P\big(\lambda(t), \theta(t)\big)}{\partial\theta(t)}\right\rvert_{(x^s, w^s)},
        \label{eq:derivative_ai_16}\\
        A^{(1,7)} &= -\frac{N_g^2}{J_t},\quad A^{(2,3)} = 1,\\
        A^{(3,1)} &= \frac{\rho\pi r^2}{2 M_t}V(t)^2 \left.\frac{\partial C_T\big(\lambda(t),\theta(t)\big)}{\partial\omega(t)}\right\rvert_{(x^s, w^s)},
        \label{eq:derivative_ai_31}\\
        A^{(3,2)} &= -\frac{K_t}{M_t},\quad A^{(3,3)} = -\frac{D_t}{M_t},\\
        A^{(3,6)} &= \frac{\rho\pi r^2}{2 M_t}V(t)^2\left.\frac{\partial C_T\big(\lambda(t),\theta(t)\big)}{\partial\theta(t)}\right\rvert_{(x^s, w^s)},
        \label{eq:derivative_ai_36}\\
        A^{(4,1)} &= -1,\ A^{(5,1)} = -\eta\cdot M_{g}^s,\ A^{(5,7)}=-\eta\cdot\omega^s.
    \end{align}
\end{subequations}
Note that, (\ref{eq:derivative_ai_11}), (\ref{eq:derivative_ai_16}), (\ref{eq:derivative_ai_31}), and (\ref{eq:derivative_ai_36}) contain derivatives of $C_P(\cdot)$ and $C_T(\cdot)$ which are deduced from their piecewise affine approximations.

\subsection{Robust Control Design}
\label{sec:uncertainty}
We now discuss the controller design for one operating region.
For this, we consider a set of $q\in\N$ constant external inputs $\mathbb{W}^s=\{w_1^s,\dots, w_q^s\}\subset\R^3$.
For each element in $\mathbb{W}^s$, an equilibrium is computed around which (\ref{eq:augmented_nonlinear_model}) is linearized.
This results in $q$ linear systems of the form
\begin{equation}
    \dot{\xi}_i(t) = A_i\xi_i(t)+B\mu_i(t),~\forall i\in\mathbb{S}=\{1,\dots,q\},
    \label{eq:linear_state_model_xi}
\end{equation}
with $\xi_i(t)\in\R^n$, $\mu_i(t)\in\R^p$, $n=7$, $p=2$.
Due to the structure of (\ref{eq:augmented_nonlinear_model}), it can be easily confirmed that each pair $(A_i, B)$ is controllable.
We aim to find a state feedback
\begin{equation}
\mu_i(t)=K\xi_i(t),
\label{eq:feedback_law}
\end{equation}
with $K\in\R^{p\times n}$, for the investigated operating region.
We introduce a cost for subsystem $i$
\begin{equation}
J_i=\frac{1}{\left\|\xi_i(0)\right\|_2^2}\int_0^\infty\xi_i(t)^TQ\xi_i(t)+\mu_i(t)^TR\mu_i(t)\dt,
\label{eq:lqr_cost_function}
\end{equation}
which is finite if ${(A_i+BK)}$ is Hurwitz.
Here, $Q\in\R^{n\times n}$ is positive semidefinite and $R\in\R^{p\times p}$ positive definite.
It is well known, (see, e.g., \cite{boyd1991linear,scherer1997multiobjective}) that the cost $J_i$ is the $H_2$-norm of the closed loop transfer function matrix ${\begin{bmatrix}Q^\frac{1}{2}& K^T{R^\frac{1}{2}}^T\end{bmatrix}^T(sI - A_i - BK)^{-1}}$, i.e.,
\begin{subequations}
    \begin{align}
        J_i &= \Tr\bigg(\begin{bmatrix}{Q^\frac{1}{2}}^T\\ R^\frac{1}{2}K\end{bmatrix}P_i\begin{bmatrix}Q^\frac{1}{2}& K^TR^\frac{1}{2}\end{bmatrix}\bigg),\\
            &= \Tr({Q^\frac{1}{2}}^TP_iQ^\frac{1}{2}) + \Tr(R^\frac{1}{2}KP_i K^TR^\frac{1}{2}),
    \end{align}
\end{subequations}
where the positive definite $P_i\in\R^{n\times n}$ satisfy
\begin{equation}
(A_i+BK)P_i+P_i(A_i+BK)^T + I_n = 0.
\label{eq:stability_condition_equal}
\end{equation}
To find a controller $K$ for the investigated operating region, we solve the following specific problem.
\begin{problem}
    \label{pro:problem}
    \begin{equation*}
    \min_{K,P}\;\Big(\underbrace{\Tr({Q^\frac{1}{2}}^TPQ^\frac{1}{2}) + \Tr(R^\frac{1}{2}KP K^TR^\frac{1}{2})}_{:=J}\Big),
    \end{equation*}
    subject to
    \vspace{-.2cm}
    \begin{subequations}\label{eq:problem_constraints}
    \begin{align}
        (A_i+BK)P+P(A_i+BK)^T + I_n &\prec 0,~\forall i\in\mathbb{S},\label{eq:stability_condition_in_problem}\\
        P&\succ 0\label{eq:p_pos_def_in_problem}.
    \end{align}
    \end{subequations}
\end{problem}
If a matrix $P$ satisfying (\ref{eq:problem_constraints}) is found, then all subsystems are asymptotically stable as $\xi(t)^TP\xi(t)$ is a joint quadratic Lyapunov function $\forall i\in\mathbb{S}$.
Moreover, subtracting (\ref{eq:stability_condition_equal}) from (\ref{eq:stability_condition_in_problem}) results in
\begin{equation}
    (A_i+BK)(P-P_i) + (P-P_i)(A_i+BK)^T \prec 0.
    \label{eq:stability_subtraction}
\end{equation}
Since $(A_i+BK)$ is Hurwitz, any solution $(P-P_i)$ to the Lyapunov inequality (\ref{eq:stability_subtraction}) is positive definite, i.e., $P\succ P_i$.
Consequently, 
\begin{multline}
    J_i = \Tr({Q^\frac{1}{2}}^TP_iQ^\frac{1}{2}) + \Tr(R^\frac{1}{2}KP_i K^TR^\frac{1}{2}) \\
    <\Tr({Q^\frac{1}{2}}^TPQ^\frac{1}{2}) + \Tr(R^\frac{1}{2}KP K^TR^\frac{1}{2}) = J,
\end{multline}
and therefore $\displaystyle\min_{P_i,K} J_i < \min_{P,K} J\quad\forall~i\in\mathbb{S},$ i.e., the minimal cost of Problem \ref{pro:problem} is an upper bound for $\displaystyle\max_{i\in\mathbb{S}}\big(\min_{P_i,K} J_i\big)$.

Due to nonlinearities, Problem~\ref{pro:problem} cannot be solved right away.
Introducing $Y=KP$ allows us to transform (\ref{eq:stability_condition_in_problem}) into an \ac{lmi} and write the cost function as
\begin{equation}
J = \Tr(QP)+\Tr(R^\frac{1}{2}YP^{-1}Y^TR^\frac{1}{2}).
\label{eq:nonlinear_trace_cost}
\end{equation}
Analogously to \cite{feron1992numerical}, we deal with the second term in (\ref{eq:nonlinear_trace_cost}) by introducing an upper bound $X\in\R^{p\times p}$ and reformulate the resulting inequality using Schur's complement \cite{boyd1994linear}, i.e.,
\begin{equation}
X\succ R^\frac{1}{2}YP^{-1}Y^TR^\frac{1}{2} \iff \begin{bmatrix}X & R^\frac{1}{2}Y\\Y^TR^\frac{1}{2} & P\end{bmatrix} \succ 0.
\end{equation}
This results in the following linear semidefinite problem.
\begin{problem}
    \label{pro:tractable_problem}
    \begin{equation*}
    \min_{P,Y,X}\;\Tr(QP)+\Tr(X),
    \end{equation*}
    subject to
    \begin{subequations}
    \begin{align}
    \begin{bmatrix}X & R^\frac{1}{2}Y\\Y^TR^\frac{1}{2} & P\end{bmatrix}\succ 0,~ P&\succ 0,\\
    A_iP + PA_i^T + BY + Y^TB^T + I_n &\prec 0,~\forall i\in\mathbb{S}.
    \end{align}
    \end{subequations}
\end{problem}
From the optimal solution to Problem~\ref{pro:tractable_problem}, the state feedback matrix can then be obtained via ${K=Y P^{-1}}$.
Note that, the resulting controller assures asymptotical stability for any system $\dot{\xi}=\tilde{A}\xi+B\mu$ with ${\tilde{A}=\sum_{i=1}^q\alpha_i A_i},{\alpha_i\geq 0}, {\sum_{i=1}^q\alpha_i=1}$ \cite{oliveira2006lmi}.

We will implement the resulting continuous-time feedback $\mu(t) = K\xi(t)$ in a discrete-time fashion, i.e., $\mu(t_k) = K\xi(t_k)$, $k\in\Nz$, $t_{k+1} - t_k=T_s$.

\subsection{External Input}
\label{sec:trajectory_generation}
The equilibria $x^s$ which are needed to derive the q linear models (\ref{eq:linear_state_model_xi}) are computed using an estimate of the external input
\begin{equation}
\hat{w}(k)=\begin{bmatrix}\hat{V}(k)& \omega^d(k)& P^d(k)\end{bmatrix}^T,
\end{equation}
Here, the wind speed estimate $\hat{V}(k)$ is typically obtained from an observer \cite{schreiber2020field}.
A high level wind farm controller provides the reference power $P^{ref}(k)\in\Rp$.
As the wind may not suffice to provide the reference power, the desired power
\begin{equation}
P^d(k) = \min\left(P^{ref}(k), \frac{\rho\pi r^2\eta}{2}\hat{V}(k)^3 C_{p,\text{opt}}\right),
\label{eq:corrected_desired_power}
\end{equation}
corrects the reference power by considering the weather-dependent maximum wind power, given at maximal power coefficient $C_{p,\text{opt}}=\max_{\lambda,\theta} C_p(\lambda, \theta)$.
From $P^d(k)$, a desired generator speed $\omega^d(k)$ is computed through a lookup table
\begin{equation}
    \omega^d(k) = \LUT\big(P^d(k)\big).
    \label{eq:setpoint_speed}
\end{equation}
This lookup table is generated analogously to \cite{kim2018design}.

\subsection{Overall Controller}
\label{sec:switching_controller_subsection}
Wind turbines usually exhibit four operating regions (see Figure~\ref{fig:regions}).
In regions 1 and 4, the turbine is not producing any power due to very low or very high wind speeds, with the latter effectively requiring a shut-down of the turbine.
In region 2, the turbine maximizes output power.
In region 3, the output power is limited to a value $P^{ref}$ below the weather-dependent available power.
The wind speed $V^d$, that separates regions 2 and 3, is a function of $P^{ref}$.
\begin{figure}[htbp]
        \centering
	    \includegraphics{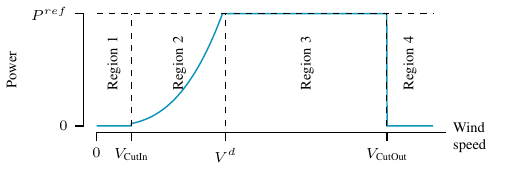}
        \caption{Wind turbine operating regions.}
        \label{fig:regions}
\end{figure}
Regions 1 and 4 are trivial from a control point of view, hence only a control scheme for regions 2 and 3 is needed.
Motivated by different control objectives in these regions, a combination of controllers is proposed.
The controller $K_j\in\R^{2\times 7}$ for operating region $j\in\{2,3\}$, is computed using the weights $Q_j\in\R^{7\times 7}$ and $R_j\in\R^{2\times 2}$ along with the sets of linear models $(A_i^j,B)\forall i\in\mathbb{S}_j$.
Here, $Q_2$ and $R_2$ are chosen to reflect the aim of maximizing produced power, while $Q_3$ and $R_3$ are chosen to achieve accurate power tracking.
The final controller is constructed as a convex combination of $K_2$ and $K_3$, i.e.,
\begin{equation}
    K = \alpha(k) K_2 + \big(1-\alpha(k)\big) K_3.
    \label{eq:controller_linear_combination}
\end{equation}
Here, $\alpha(k)\in[0,1]$ is computed based on the wind speed estimate $\hat{V}$ and $V^d$, i.e.,
\begin{equation}
\alpha = \begin{cases}
1,&\text{if } \hat{V}\leq V^d-\Delta V,\\
0,&\text{if } \hat{V}\geq V^d+\Delta V,\\
-\frac{1}{2\Delta V}(\hat{V} - V^d - \Delta V),&\text{else}.
\end{cases}
\label{eq:alpha_computation}
\end{equation}
The linear ramp between the thresholds enables a smooth transition between $K_2$ and $K_3$.


\section{Case Study}
\label{sec:case_study}
In this section, we evaluate the proposed robust \ac{lq} controller (\ref{eq:controller_linear_combination}), (\ref{eq:alpha_computation}) using OpenFAST, an open-source wind turbine simulation tool \cite{jonkman2022openfast}.
OpenFAST provides a detailed model describing the relation between $\begin{bmatrix}\theta& M_g\end{bmatrix}^T$ and $\begin{bmatrix}\omega& x_T& v_T\end{bmatrix}^T$.
We augment this simulation model with additional states as described in Section~\ref{sec:augmentation} (see Figure~\ref{fig:augmented_model}).
Furthermore, pitch angle, generator torque and their derivatives w.r.t. time are subject to saturation constraints, i.e., 
\begin{subequations}
    \begin{equation}
        \theta(t)=\sat_{\underline{\theta}}^{\overline{\theta}}\bigg(\int_0^t\sat_{\underline{\Delta\theta}}^{\overline{\Delta\theta}}u_1(\tau)\text{d}\tau\bigg),
    \end{equation}
    \begin{equation}
        M_g(t)=\sat_{\underline{M}_g}^{\overline{M}_g}\bigg(\int_0^t\sat_{\underline{\Delta M}_g}^{\overline{\Delta M}_g}u_2(\tau)\text{d}\tau\bigg),
    \end{equation}
\end{subequations}
where $\underline{\theta}, \overline{\theta}, \underline{M}_g, \overline{M}_g\in\R$ are the lower and upper bounds for pitch angle and generator torque, respectively.
Moreover, $\underline{\Delta\theta}, \overline{\Delta\theta}, \underline{\Delta M}_g, \overline{\Delta M}_g\in\R$ denote the lower and upper bounds of the pitch angle rate of change and generator torque rate of change, respectively.
We compare the controller presented in this paper (Robust LQ) with the baseline controller (BLC) from \cite{schlipf2016lidar} and the \ac{lq} optimal controller described in a previous work \cite{grapentin2022lq}.
The latter differs from the one considered in this paper, as robustness issues and tower movement were not explicitly considered.
Moreover, no integral error for power reference tracking was included.
For all simulations, a sampling time of $T_s=\SI{4}{\milli\second}$ is used.
The wind speed estimates for all control approaches were generated employing the wind observer presented in \cite{schreiber2020field}.
\begin{figure}[htbp]
	\centering
	\includegraphics{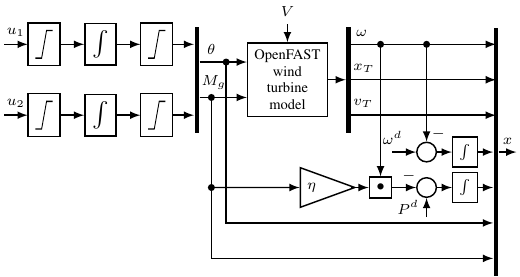}
	\caption{Wind turbine simulation model.}
	\label{fig:augmented_model}
\end{figure}

In Figure~\ref{fig:powerRegions3and2}, the Robust LQ and the BLC were compared in decelerating wind conditions.
During the first \SI{10}{\minute}, the wind speed allows arbitrary power tracking, therefore mostly controller $K_3$ is used.
After $t=\SI{10}{\minute}$, controller $K_2$ is used to maximize the power production.
While maximizing power, Robust LQ and BLC behave similarly, with slightly higher oscillations for the Robust LQ controller.
\begin{figure}[b]
	\centering
	\includegraphics{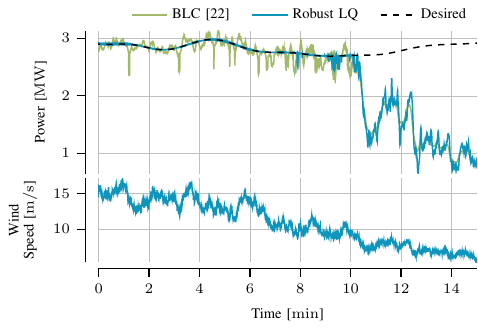}
	\caption{Power tracking comparison.}
	\label{fig:powerRegions3and2}
\end{figure}

\begin{figure}[t]
	\centering
	\includegraphics{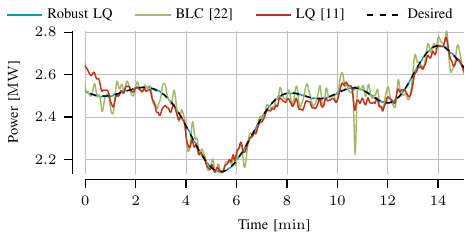}
	\caption{Power tracking comparison.}
	\label{fig:powerTracking}
\end{figure}
Figure~\ref{fig:powerTracking} displays the desired power as well as the output power for the Robust LQ, the BLC, and the \ac{lq} optimal controller from \cite{grapentin2022lq}.
The results in Figure~\ref{fig:powerTracking} were generated with an average wind speed of \SI{14.8}{\meter\per\second} and a turbulence intensity of \SI{6}{\percent}.
Evidently, the Robust \ac{lq} controller tracks the desired output power much more accurately: the root-mean-square power tracking error is reduced by more than \SI{80}{\percent} compared to the other controllers.
This is achieved while keeping the rates of change within the desired range: the associated max. generator torque rate of change is at \SI{0.364}{\kilo\newton\meter\per\second}, which is far from the limits of $\pm\SI{15}{\kilo\newton\meter\per\second}$.
Moreover, the maximum absolute pitch angle rate of change is \SI{4.55}{\degree\per\second}, which is within the limits of $\pm\SI{7}{\degree\per\second}$.
Hence, accurate power tracking does not strain the actuators to their respective limits.

\begin{figure}[htbp]
    \centering
	\includegraphics{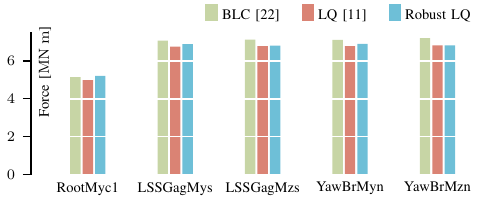}
    \caption{Damage equivalent loads for the three controllers. The blade root out-of-plane bending moment (RootMyc1), shaft non-rotating out-of-plane bending moment (LSSGagMys), shaft non-rotating yaw bending moment (LSSGagMzs), tower top fore-aft bending moment (YawBrMyn), and tower top torsion bending moment (YawBrMzn).}
    \label{fig:damageEquivalentLoads}
\end{figure}
In Figure~\ref{fig:damageEquivalentLoads}, the \acp{del} are displayed for the same scenario as shown in Figure~\ref{fig:powerTracking}.
The damage amplitudes are computed using the rainflow counting algorithm \cite{dowling1971fatigue}.
From these, \acp{del} are computed by extrapolating the damage amplitudes during the simulation onto a wind turbine life time \cite{burton2011wind}.
Evidently, the robust LQ controller achieves similar \acp{del} as the BLC and the \ac{lq} optimal controller from \cite{grapentin2022lq}.
While the expectation was to further decrease the tower top fore-aft bending moment (YawBrMyn) by modeling the tower top fore-aft movement, this was not achieved.
However, the greatly improved power tracking capabilities provided by the novel controller do not increase mechanical wear.


\section{Conclusion}
\label{sec:conclusion}
In this paper, we introduced a robust \ac{lq} optimal control approach employing \acp{lmi} for active power tracking of individual wind turbines.
We augmented the model to meet different control objectives.
This includes an integral error for power tracking and differential inputs.
In a realistic numerical case study, we showed that using the controller leads to highly accurate power tracking performance without compromising \acp{del}.
Future work will focus on including yaw dynamics into the control scheme, such that all actuators of the wind turbine are controlled by a multivariable controller.
Moreover, a validation in a wind tunnel is planned.


\addtolength{\textheight}{-0cm}   




\bibliography{literature}

\end{document}